\documentclass[11pt]{article}
\usepackage{hyperref}
\usepackage[round]{natbib}
\usepackage{authblk}
  \usepackage{pdfpages}
  \usepackage{amsmath,amssymb,amsfonts,mathrsfs,amsthm}
\usepackage{hyperref}

  \usepackage{pdfpages}
  \usepackage{amsmath,amssymb,amsfonts,mathrsfs,amsthm}

\usepackage[round]{natbib}
\usepackage{authblk}

\usepackage{bbm}

\usepackage{epsf}
\usepackage{epsfig}
\usepackage{setspace}
\usepackage{multirow}
\usepackage{subfigure}
\usepackage{comment}
\usepackage{gensymb}

\usepackage{graphics}
\usepackage{graphicx}
\usepackage{geometry}
\usepackage{array}
\usepackage{float}
\usepackage{lscape}
\usepackage{color}
\usepackage{calligra}
\usepackage{dsfont}
\usepackage{tabularx}
\usepackage{dcolumn}
\usepackage[bottom]{footmisc}
\usepackage{enumerate}
\usepackage{url}
\usepackage[normalem]{ulem}
\usepackage{epstopdf}
\usepackage[space]{grffile}
\usepackage{setspace}

\usepackage{accents}

\newtheorem{proposition}{Proposition}
\newtheorem{theorem}{Theorem}
\newtheorem{lemma}{Lemma}

\newtheorem{example}{Example}
\newtheorem{remark}{Remark}

\usepackage{lineno}
\usepackage{lipsum}

\usepackage{mathtools}

\DeclareMathAlphabet{\mathpzc}{OT1}{pzc}{m}{it}
\DeclareMathAlphabet{\mathcalligra}{T1}{calligra}{m}{n}
\title{On Sybil-proof Mechanisms}
\author[1]{Minghao Pan}
\author[2]{Bruno Mazorra}
\author[3]{Christoph Schlegel}
\author[4]{Akaki Mamageishvili}

\affil[1]{California Institute of Technology}
\affil[4]{Offchain Labs}
\affil[1,2,3]{Flashbots}

\date{}

\title{On Sybil-Proof Mechanisms} 

\newtheorem{question}{Question}

\begin{document}


\maketitle
\begin{abstract}
We show that in the single-parameter mechanism design environment, the only non-wasteful, symmetric, incentive compatible and Sybil-proof direct mechanism is a second price auction with symmetric tie-breaking. Thus, if there is private information, lotteries or other mechanisms that do not always allocate to a highest-value bidder are not Sybil-proof or not incentive compatible. Moreover, we show that our main (im)possibility result extends beyond linear valuations, but not to multi-unit object allocation with unit-demand bidders. 

We also provide examples of mechanisms (with higher interim payoff for the bidders than a second price auction) that satisfy all of the other axioms and a weaker, Bayesian notion of Sybil-proofness. Thus, our (im)possibility result does not generalize to the Bayesian setting and we have a larger design space: With Sybil constraints, equivalence between dominant strategy and Bayesian implementation (that holds in classical single-parameter mechanism design without Sybils) no longer holds. 
\end{abstract}


\section{Introduction}

In environments where creating new identities is cheap or free and verifying participants' identities is hard, mechanism design must consider the possibility that a participant creates multiple identities to secure a better outcome from the mechanism. Such manipulation of a mechanism is commonly referred to as a Sybil attack and a mechanism that is immune to such attacks is called \emph{Sybil-proof}.
Sybils are of particular concern in mechanisms deployed online, in case identities are not verified. Blockchain systems, for example, where participation is permissionless, are an environment that are both rich in deployed mechanisms and susceptible to Sybil-attacks.  In principle, however, Sybils can also occur in offline environment, e.g. if a participant asks a friend to participate in a mechanism on their behalf.

When allocating private goods, a natural choice of a Sybil-proof mechanism is an auction that assigns (all units of) the good to the highest-value bidder(s). However, pure auctions may be undesirable if we aim to ensure fairness, wider participation, or other distributional goals.

In the case of online sales of event tickets,  tickets are often sold at a relatively low price where there is still excess demand. The ``under-pricing" of tickets is a way for the event organizers to give dedicated fans with smaller budgets a chance to participate, yet this also makes the primary sale a target for middlemen that resell tickets with a premium on a secondary market~\cite[]{budish2023}. If quantities of tickets per buyer are capped, middlemen effectively achieve this by using Sybils in the primary sale.\footnote{Related issues also appear in crypto-currency ``airdrops" where protocols want to reward early adopters with tokens, see~\cite[]{messias2023}, and the use of non-proportional rules, has lead to wide-spread exploits through Sybils who afterwards sell their airdrop on a secondary market.} 

As another example, in blockchain systems, randomness in allocating the right to propose blocks is desirable to avoid power concentration, which could lead to censorship and undermine decentralization. Bitcoin's proof-of-work mechanism in particular, which uses pseudorandomness in proposer selection, has been argued to be an effective and Sybil-proof method to achieve decentralization of the system by having many different actors ``mining" blocks. 
This assessment can, however, change if there is significant heterogeneity in value for proposal rights.\footnote{Heterogeneous (and private) value is a concern, as different participants can generate different values from block proposal rights. For example, some block proposers have exclusive access to some of the submitted transactions to include in their block, as documented in~\cite[]{oz} for the case of Ethereum. In practice, this had led to the creation of out-of-protocol secondary market for the content of blocks, and these markets exhibit a high degree of market concentration in the hand of few block builders~\cite[]{oz}; in other words, empirically the assertion of a ``monopoly without monopolist"~\cite[]{huberman2021} does not really hold.} 

Given these considerations, it is natural to ask whether there are Sybil-proof mechanisms for assigning private goods other than auctions. Our main result, Theorem~\ref{maintheorem}, gives a strong negative answer to this question: in the classical Myersonian mechanism design setting with quasi-linear preferences and one-dimensional types, we show that the only monotonic, symmetric, and non-wasteful allocation rule for which the induced mechanism with ``Myerson payments'' is Sybil-proof, is the rule that allocates everything to the highest value bidder, breaking ties uniformly (if the good is indivisible), resp. sharing the unit equally among highest value bidders (if the good is divisible). Thus, while there is a large design space of monotonic allocation rules - as two extreme cases we could assign the good with equal probability, independently of value, among participants or we could always assign the good to the highest value bidder - adding Sybil-proofness collapses the set of implementable allocation rules to a single one.  The result extends beyond linear valuation (Proposition~\ref{thm:superadditive}), but interestingly not to the multi-unit case with unit-demand bidders (Proposition~\ref{MultiUnit}).

While our main theorem is formulated for direct mechanisms that elicit values from the participants, indirect mechanism cannot circumvent the impossibility, as we have a version of the revelation principle in the presence of Sybils (Section~\ref{sec:revelation_sybils}): for every symmetric indirect mechanism with Sybils that implements an objective function in dominant strategies, the direct mechanism induced by the objective function is Sybil-proof and incentive compatible.

One way out of our impossibility result is to move to the Bayesian setting. Instead of the strong ex-post requirement of Sybil-proofness, we could demand that no bidder can gain from introducing additional Sybils in equilibrium, provided all other bidders truthfully report their valuations and do not use Sybils. In this weaker setting, we establish that there are indeed mechanisms which are non-wasteful, symmetric, incentive-compatible, and Bayesian Sybil-proof —yet they do not always allocate to the highest bidder (Theorem~\ref{thm:general}).
The mechanism allocates the object through a lottery among high value bidders above a threshold (and through an auction if there are no high value bidders).
Thus, under Sybil-constraints and in contrast to the classical mechanism design setting without Sybils, Bayesian implementation is strictly weaker than dominant strategy implementation and allows for a larger design space.

\subsection{Related Work}
Sybil-proof mechanisms have been studied in various contexts, including voting systems~\cite[][]{wagman2008}, combinatorial auctions~\cite[][]{false_name,false_name_vcg_bayesian,false_name_combinatorial}, recommender systems~\cite[][]{brill2016} and blockchain systems~\cite[][]{chen2019,Leshno2020}, and different Sybil-proof mechanisms\footnote{Previously literature also used the term {\it false-name proof} mechanisms for a stronger notion than the one we use.} have been discussed in these contexts. Most related to our work, \cite[][]{false_name} prove that there is no Pareto-efficient, false-name-proof combinatorial auction. We do not require Pareto efficiency, as we are primarily interested in whether non-trivial distributional goals can be achieved. However, it is as a consequence of our theorem that the combination of Non-wastefulness, Sybil-proofness and Incentive Compatibility implies Pareto-efficiency.

Our results also relate to the literature on mechanism design with one-dimensional types. In the case without Sybils, for one-dimensional types Dominant-Strategy- and Bayesian- incentive compatibility are equivalent in the sense that for each Bayes-Nash IC mechanism one can construct a Dominant-Strategy IC mechanism that gives each participant the same interim expected pay-off~\cite[]{ECTA1,ECTA2}. We show that the statement is no-longer true with additional Sybil-proofness constraints and dominant-strategy implementation is strictly stronger than Bayes-Nash implementation in the presence of Sybils.  

Similar but different impossibility results also occur if one considers collusion resistance instead of Sybil resistance as a desideratum.\footnote{Collusion differs from sybilling in so far as different agents colluding, in general, have different preferences over outcomes of a mechanism. Moreover, collusion is usually studied in a fixed population setting, whereas sybilling means creating (an arbitrary number of) new participants with alligned preferences.} For transaction fee mechanisms for blockchains for example,~\cite[]{transaction_fee_mechanisms} show that there is no non-trivial mechanism that satisfies constant probability allocation, incentive compatibility, miner incentive compatibility and collusion resistance.

\section{Model}
We consider mechanisms with a variable population of bidders: A {\bf (direct) mechanism} \footnote{Focusing on direct mechanism, where participants report their value to the mechanism which determines an allocation from the reports is without loss of generality by the revelation principle, which we formally establish in Section~\ref{sec:revelation_sybils}.} 
specifies for each finite $\mathcal{N}\subset\mathbb{N}$ set of {\bf bidders} (some of them real agents, some of them, possibly, Sybils), an {\bf allocation rule} $x^{\mathcal{N}}:\mathbb{R}_+^{\mathcal{N}}\rightarrow\Delta(\mathcal{N})$ where $\Delta(\mathcal{N}):=\{x\in\mathbb{R}_+^{\mathcal{N}}:\sum_{i\in\mathcal{N}}x_i\leq 1\}$  and a {\bf payment rule} $p^{\mathcal{N}}:\mathbb{R}^{\mathcal{N}}_+\rightarrow\mathbb{R}^\mathcal{N}.$
In the following, we will often omit the superscript $\mathcal{N}$ when there is no ambiguity.
We can interpret the allocation shares $x^{\mathcal{N}}_i(v)$ for the reported {\bf values}\footnote{We work with the non-negative reals as type space. Alternatively, we could also work with an interval $[0,\bar{v}]$ as type space and obtain the same characterization result with completely analogous proofs.} $v\in\mathbb{R}_+^{\mathcal{N}}$ either as probabilities of obtaining an indivisible good or as an allocation of a perfectly divisible good of which one unit is distributed in total. 
A bidder $i$ with value $v_i\geq 0$ has a linear utility $$U_i=v_ix_i-p_i$$
 if allocated a share $x_i$ and making a payment of $p_i$.


Next, we introduce several axioms that the mechanisms should satisfy. 
First, we want participation to be voluntary, in the sense that bidders who do not derive value from the item do not need to pay:\newline

\noindent {\bf Payment normalization}:\label{zero_bid_payment}
For each finite $\mathcal{N}\subseteq\mathbb{N}$, bidder $i\in\mathcal{N}$ and values $v_{-i}\in\mathbb{R}_+^{\mathcal{N}\setminus\{i\}}$ we have 
$$
p_i^{\mathcal{N}}(0,v_{-i})\leq0.
$$

Second, we want the allocation rule to always allocate the whole unit:
\newline

\noindent {\bf Non-wastefulness}: For each finite $\mathcal{N}\subset\mathbb{N}$ set of agents and $v\in \mathbb{R}^{\mathcal{N}}_+$  we have 

$$
    \sum_{j\in\mathcal{N}}x_j^{\mathcal{N}}(v)=1.
$$

Third, we want the allocation rule to treat agents symmetrically:\newline

\noindent {\bf Symmetry}: For each finite $\mathcal{N}\subset\mathbb{N}$ set of agents, $v\in \mathbb{R}^{\mathcal{N}}_+$, permutation $\pi:\mathcal{N}\rightarrow\mathcal{N}$ and each $j\in\mathcal{N}$ we have $$x_j^{\mathcal{N}}(v)=x_{\pi(j)}^{\mathcal{N}}(\{v_{\pi(i)}\}_{i\in\mathcal{N}}).$$

Fourth, we want the mechanism to be dominant strategy\footnote{It is well-known that for the single-parameter setting a Bayesian incentive compatible mechanism exists if and only if a dominant strategy incentive compatible mechanism exists for the same allocation rule so that dominant strategy incentive compatibility is not really a stronger property.} incentive compatible:\newline

\noindent {\bf Incentive Compatibility}: For each $v\in\mathbb{R}_+^{\mathbb{N}}$, each finite $\mathcal{N}\subset\mathbb{N}$ set of agents, each agent $i\in\mathcal{N}$ and bid $u_i\geq 0$ we have
$$v_ix_i^{\mathcal{N}}(v)-p_i^{\mathcal{N}}(v)\geq v_ix_i^{\mathcal{N}}(u_i,v_{-i})-p_i^{\mathcal{N}}(u_i,v_{-i}).$$
As known from classical results \cite[]{myerson1981}, if zero value bidders do not pay, this is equivalent to using ``Myerson payments'',
\begin{align}p_{j}(v):=v_{j}\cdot x_{j}(v_{j},v_{-j})-\int_{0}^{v_{j}}x_{j}(z,v_{-j})dz,\label{Myerson}
\end{align}
and requiring the axiom of 
\newline\\
\noindent{\bf Monotonicity}: For each finite $\mathcal{N}\subset\mathbb{N}$ set of agents, the function $x^{\mathcal{N}}$ is non-decreasing on its domain.\\

Fifth, we want the mechanism to be Sybil-proof. Previous literature has used the term ``false-name-proofness",~\cite[]{false_name}, but usually for the combination of incentive compatibility and Sybil-proofness, which requires immunity to deviations where the bidder reports a different value \emph{and} creates Sybils. For our result a weaker notion of Sybil-proofness is needed, which only requires immunity to Sybil attacks where a bidder reports truthfully from his original account and creates one Sybil with an arbitrary bid:\\

\noindent {\bf Sybil-proofness}: For each $v\in\mathbb{R}^{\mathbb{N}}_+$, each finite $\mathcal{N}\subset\mathbb{N}$, $i\in\mathcal{N}$, $j\in\mathbb{N}\setminus\mathcal{N}$  and bid $u\geq0$, we have

\[v_ix_i^{\mathcal{N}}(v)-p_i^{\mathcal{N}}(v)\geq v_i\left(x_i^{\mathcal{N}\cup \{j\}}(v,u)+x_j^{\mathcal{N}\cup \{j\}}(v,u)\right)-p_i^{\mathcal{N}\cup \{j\}}(v,u)-p_j^{\mathcal{N}\cup \{j\}}(v,u).\]

\begin{remark}
There are other closely related notions of (ex-post) Sybil-proofness that we could use instead and still obtain our main result. Note that, as we require immunity to manipulations through the creation of one Sybil, the following main result will hold a forteriori also for stronger notions of Sybil-proofness that require immunity to manipulation through the creation of an arbitrary number of Sybils. In particular, our main result holds with Sybil-Proofness (and incentive compatibility) replaced by the False-Name-Proofness axiom studied in the previous literature. However, we use our version of the Sybil-proofness axiom to establish an impossibility result with as weak of an ex-post Sybil-proofness axiom as possible yielding a stronger result.

Another natural variant of Sybil-proofness would be to require that for each bid vector representing multiple bids of one agent, there is a single bid that weakly dominates the Sybil-bid.
With this notion, our main result would again hold true, as Incentive Compatibility and this version of Sybil-Proofness imply again False-Name-Proofness which implies our version of Sybil-Proofness (and Incentive Compatibility).
\end{remark}
\section{Characterization Result}
We now show that the only direct mechanism that satisfies all of the above axioms is a second-price auction.

\begin{theorem}\label{maintheorem}
A direct mechanism is payment normalized, non-wasteful, symmetric, incentive compatible, and Sybil-proof if and only if it is a second price auction with symmetric tie-breaking.
\end{theorem}
\begin{proof}
Note that payment normalization and Sybil-proofness imply that $p_{j}(0,v_{-j})=0$ for a bidder $j$ who bids $0$. Thus, by Myerson's lemma, the payments that implement the rule in dominant strategies are defined by Equation (\ref{Myerson}).
Subsequently, we will use the following two facts about the payments: first we have individual rationality under Myerson payments
\[
p_{j}(v)\leq v_{j}\cdot x_{j}(v),
\]
for any $v\in \mathbb{R}_+^\mathcal{N}$,
and second the payoff of a bidder $j$  whose value is $v_{j}$ and bids
truthfully when the other bidders bid $v_{-j}$ is%
\begin{eqnarray}
U_j(v):=v_{j}\cdot x_{j}(v_{j},v_{-j})-p_{j}(v)=\int_{0}^{v_{j}}x_{j}(z,v_{-j})dz.
\label{eqn0}
\end{eqnarray}

In the following, we show three lemmas that will be used for the proof of the theorem. The first lemma says that, when there are many bidders that bid the same value, one bidder bidding higher will almost certainly get the good.

\begin{lemma}
\label{lm: limit of x1}For any $u>v$, $u,v\in \mathbb{R}_+$, we have
\[
\limsup_{n}x_{1}^{1\cup \{2,\ldots ,n\}}(u,v_{[2,n]})=1.
\]
\end{lemma}

\begin{proof}
Suppose that the value of Bidder $1$ is $u$ and that the values of all other $n-1$ bidders are $v$. Bidder $1$ could deviate by bidding $u$ from his
original account and creating a Sybil that bids $v$. By Symmetry and Non-wastefulness, the Sybil
account has a chance of 
\[
x_{n+1}(u,v_{[2,n+1]})=\frac{1}{n}\left[ 1-x_{1}(u,v_{[2,n+1]})\right] 
\]%
to win the lottery. The payment from the Sybil account is at most $v\cdot
x_{n+1}(u,v_{[2,n+1]})$, by individual rationality. In order to prevent Bidder $1$ from this deviation,
we must have%
\begin{eqnarray}
U_{1}(u,v_{[2,n]}) &\geq &U_{1}(u,v_{[2,n+1]})+(u-v)x_{n+1}(u,v_{[2,n+1]}) 
\nonumber \\
&=&U_{1}(u,v_{[2,n+1]})+(u-v)\frac{1}{n}\left[ 1-x_{1}(u,v_{[2,n+1]})\right]
,  \label{eqn2}
\end{eqnarray}%
where the first inequality is implied from~\eqref{eqn0} and properties of the integral.

By a straightforward induction on $n$, we have 
\[
U_{1}(u,v)\geq (u-v)\sum_{n=2}^{\infty }\frac{1}{n}\left[
1-x_{1}(u,v_{[2,n+1]})\right] .
\]%
Suppose by contradiction that $\limsup_{n}x_{1}(u,v_{[2,n]})<1$, then there
exists $N$ large and $a>0$ small such that for any $n>N$, $%
x_{1}(u,v_{[2,n]})<1-a$. Hence, 
\begin{eqnarray*}
U_{1}(u,v) &\geq &(u-v)\sum_{n=2}^{\infty }\frac{1}{n}\left[
1-x_{1}(u,v_{[2,n+1]})\right]  \\
&\geq &(u-v)\sum_{n=N}^{\infty }\frac{1}{n}\left[ 1-x_{1}(u,v_{[2,n+1]})%
\right]  \\
&\geq &(u-v)\sum_{n=N}^{\infty }\frac{a}{n} \\
&=&\infty ,
\end{eqnarray*}%
which is impossible. We can conclude that $\limsup_{n}x_{1}(u,v_{[2,n]})=1.$
\end{proof}

We then lower bound the expected payoff of the higher value bidder when there are
two bidders.

\begin{lemma}
\label{U lower bound} $U_{1}(u,v)\geq u-v$ if $u>v$, $u,v\in \mathbb{R}_+$.
\end{lemma}

\begin{proof}
Fix $\varepsilon \in (0,u-v)$ small. Equation (\ref{eqn2}) implies, $%
U_{1}(u,v)\geq U_{1}(u,v_{[2,n]})$ for any $n\geq 2$ and Equation (\ref{eqn0}) together with the monotonicity of $x$ implies
\begin{eqnarray*}
U_{1}(u,v_{[2,n]}) &=&\int_{0}^{u}x_{1}(z,v_{[2,n]})dz \\
&\geq &\int_{v+\varepsilon }^{u}x_{1}(z,v_{[2,n]})dz \\
&\geq &(u-v-\varepsilon )x_{1}(v+\varepsilon ,v_{[2,n]}).
\end{eqnarray*}%
Taking $\limsup_n $ on both sides and using Lemma~\ref{lm: limit of x1}, we get%
\[
U_{1}(u,v)\geq \limsup_n U_{1}(u,v_{[2,n]})\geq (u-v-\varepsilon ).
\]%
As $\varepsilon >0$ is arbitrary, we arrive at%
\[
U_{1}(u,v)\geq u-v.
\]
\end{proof}

Next, we show that the utility of a bidder is (weakly) decreasing in the bid of the other bidder.

\begin{lemma}\label{U cts}
For a fixed $u\geq 0$, the function $v\longmapsto U_{1}(u,v)$ is decreasing in $[0,u]$.
\end{lemma}

\begin{proof}
Let $v_{1}>v_{2}$. Then by Equation (\ref{eqn0}) and monotonicity,  
\[
U_{1}(u,v_{1})=\int_{0}^{u}x_{1}(z,v_{1})dz=%
\int_{0}^{u}(1-x_{2}(z,v_{1}))dz
\leq%
\int_{0}^{u}(1-x_{2}(z,v_{2}))dz=U_{1}(u,v_{2}).
\]
\end{proof}

We then fix $u$ as the value of Bidder $1$ and draw the value of Bidder $2$
from the uniform distribution between $0$ and $u$. The expected utility of Bidder $1$ is%
\[
\frac{1}{u}\int_{v=0}^{u}U_{1}(u,v)d v=\frac{1}{u}\int_{v=0}^{u}%
\int_{z=0}^{u}x_{1}(z,v)dzdv=\frac{1}{u}\int \int_{u\geq v>z\geq
0}(x_{1}(z,v)+x_{1}(v,z))=\frac{u}{2},
\]%
where we use Equation (\ref{eqn0}) in the first equality and Non-wastefulness in the third equality. 
On the other hand, by Lemma \ref{U lower bound}, the expected utility of Bidder $1$
is at least%
\[
\frac{1}{u}\int_{v=0}^{u}U_{1}(u,v)dv\geq \frac{1}{u}\int_{v=0}^{u}(u-v)dv=%
\frac{u}{2}.
\]%
Thus, for $v\in (0,u)$, $U_{1}(u,v)=u-v$, almost surely. In particular, there exists a sequence $%
v_{i}\uparrow u$ such that $U_{1}(u,v_{i})=u-v_{i}$. By Lemma \ref{U cts}, we have%
\[
U_{1}(u,u)\leq \lim_{i}U_{1}(u,v_{i})=0
\]%
so that $U_{1}(u,u)=0$. However, we know%
\[
0=U_{1}(u,u)=\int_{z=0}^{u}x_{1}(z,u)dz
\]%
so that $x_{1}(z,u)=0$ almost surely. As $z\mapsto x_1(z,u)$ is non-decreasing due to monotonicity, we have $x_1(z,u)=0$ for every $z<u$. We have established that for the case of two bidders, the allocation rule always assigns the item to the highest value bidder. Next, we show by induction on the number of bidders that in the case of more than two bidders the item is also allocated to the highest value bidder:

We claim that for any $n\geq  2$ and any $v<u:=\max\{u_2,\ldots,u_n\}$, Bidder $1$ will not obtain the good if reporting $v$, i.e. $x_1(v,u_2,\ldots,u_n)=0$.  We proceed by induction on $n$. We already know that the base case $n=2$ is true. Suppose that the claim is true for $n=k$ and assume towards contradiction that $x_1(v,u_2,\ldots,u_{k+1})>0$ for some $v<u:=\max\{u_2,\ldots,u_{k+1}\}$. By relabeling if necessary, assume that $u_2\geq u_3\geq \cdots \geq u_{k+1}$. Consider the scenario that there are $k$ bidders where Bidder 1 has value $u$ and Bidder $i$ has value $u_i$ for $2\leq i\leq k$. Then $u=\max\{u_2,\ldots,u_{k+1}\}=\max\{u_2,\ldots,u_{k}\}$. If Bidder $1$ bids truthfully, his utility payoff is 
\[\int_0^u x_1(z,u_2,\ldots,u_{k})dz=0\]
by the induction hypothesis. However, if Bidder $1$ deviates by bidding $u$ himself and creating a Sybil that bids $u_{k+1}$, then his utility payoff is 
\begin{equation*}
    \begin{split}
       & \;\;\;\;\int_0^{u}x_1(z,u_2,\ldots,u_{k+1})dz+\left(u\cdot x_{k+1}(u,u_2,\ldots,u_{k+1})-p_{k+1}(u,u_2,\ldots,u_{k+1})  \right)\\ 
& \geq \int_v^{u}x_1(v,u_2,\ldots,u_{k+1})dz+\left(u_{k+1}\cdot x_{k+1}(u,u_2,\ldots,u_{k+1})-p_{k+1}(u,u_2,\ldots,u_{k+1})  \right)\\ 
       & \geq  (u-v)x_1(v,u_2,\ldots,u_{k+1})+U_{k+1}(u,u_2,\ldots,u_{k+1})\\
       &>   0,
    \end{split}
\end{equation*}
which contradicts Sybil-proofness.

To summarize, we have shown the allocation rule shall reward the item to the highest bidder (when there is a tie among bids, the symmetry assumption implies uniform tie-breaking rule) and therefore, by incentive compatibility, the mechanism is a second-price auction. 
\end{proof}
Next, we will see that the axioms in our characterization are logically independent, and dropping any of the axioms give additional mechanisms.
If we drop Sybil-proofness, then any Symmetric, Monotone, and Non-wasteful allocation rule with ``Myerson payments'' will satisfy all the axioms. 
A second-price auction with a participation cost that does not depend on (or increases with) the number of bidders satisfies all axioms but payment normalization.
The generalized proportional mechanism
$$x_i^{\mathcal{N}}(u)=\frac{f(u_i)}{\sum_{j\in \mathcal{N}} f(u_j)},\quad p_i(u)=cu_i,$$
for a convex increasing function $f$ with $f(0)=0$ and $c>0$ satisfies all the axioms but Incentive Compatibility. If we drop Non-wastefulness, then the second price auction with a reservation price that does not depend on (or increases with) the number of bidders will satisfy all other axioms. We give another interesting example which does not satisfy Non-wastefulness and satisfies all the other axioms in the following.
\begin{example}\label{example}
Consider the allocation rule 
\[x_i^{\mathcal{N}}(u)=2^{-|\mathcal{N}|}\frac{u_i}{\sum_{j\in \mathcal{N}} u_j}\]
and the payment rule given by the ``Myerson payments''.
Then the mechanism satisfies Payment Normalization, Symmetry, Incentive Compatibility, Monotonicity, and Sybil-proofness. The only nontrivial property to check is Sybil-proofness. 

Suppose that there are $N$ bidders with values $u$. If they all bid truthfully and we define $C:=\sum_{j>1}u_j$, Bidder $1$ has utility
payoff,%
\[
U_{1}=2^{-N}\int_0^u\frac{z}{z+C}dz\]%
Suppose that Bidder $1$ deviates by creating a
Sybil bidding $x$. Then his utility payoff will be%

\begin{eqnarray*}
    2^{(N+1)}U^{\prime }_1&=&\int_0^u\frac{z}{z+x+C}dz+\frac{x(u_1-x)}{u_1+x+C}+\int_0^x\frac{z}{z+u_1+C}dz
\end{eqnarray*}

We show that $U_1'\leq U_1$ for any $x\geq0$, i.e. that bidder $1$ has no incentive to deviate by creating Sybils. We have  
\begin{align*}
2^{N+1}(U_1-U_1')&=2\int_0^{u_1}\frac{z}{z+C}dz-\int_0^{u_1}\frac{z}{z+x+C}dz-\frac{x(u_1-x)}{u_1+x+C}-\int_0^x\frac{z}{z+u_1+C}dz\\&\geq\int_0^{u_1}\frac{z}{z+C}dz-\frac{x(u_1-x)}{u_1+x+C}-\int_0^x\frac{z}{z+u_1+C}dz
\end{align*}
Note that the expression on the right hand side of the inequality is positive on the boundary i.e. at $x=0$ and at $x=\infty$ (as it diverges to $\infty$). The expression on the right hand side of the inequality has one critical point at $x=u_1$, and the expression is positive at that critical point. Therefore the difference in utility is bounded from below with a strictly positive expression. 
\end{example}
\section{Bayesian Sybil-Proofness}\label{Bayes}
In the previous section, we introduced a strong notion of Sybil-proofness, which ensures that no agent can benefit from creating Sybils, regardless of the strategies employed by other participants. However, we may wish to relax this requirement and consider a weaker, \emph{Bayesian} notion of Sybil-proofness, where agents hold prior beliefs about the number of other participants and/or their valuations. In this Bayesian setting, we require that creating Sybils is not profitable \emph{in equilibrium}, that is, if no agent creates Sybils and all other participants truthfully report their valuations, then no agent can increase their payoff by creating Sybils instead of using a single identity.

In what follows, we show that there exist mechanisms that satisfy non-wastefulness, symmetry, incentive-compatibility, and a Bayesian notion of Sybil-proofness, yet do not allocate the item to the highest-value bidder with probability one. 
The mechanism allocates the object through a lottery among high value bidders above a threshold set ex-ante (and through an auction if there are no high value bidders). In the mechanism, the probability of not allocating to the highest value bidder conditional on the number of bidders being $n$ is bounded from below by a constant independent of $n$. Moreover, the mechanism generates higher interim payoff for bidders than the second price auction.

Technically, for the Bayesian version of our model we need to enhance it by a prior which is a probability distribution $\mu$ over $\bigcup_{\mathcal{N}\subset\mathbb{N},|\mathcal{N}|<\infty}\mathbb{R}_+^{\mathcal{N}} $. Thus, both, valuations and the number of real bidders, can be unknown. Denote by $\mu(\cdot|v_i)$ the posterior of bidder $i$ if he is part of the mechanism and has type $v_i$. The posterior is a probability distribution over $\bigcup_{i\in\mathcal{N}\subset\mathbb{N}}\mathbb{R}_+^{\mathcal{N}} .$
 Subsequently, we focus on the independent private value setting, where the conditional distribution of values, conditional on the set of agents in the mechanism being $\mathcal{N}$, is $F^{\mathcal{N}}$ for a distribution $F$ on $\mathbb{R}_+$ that is twice differentiable and strictly increasing on its support (assumed to be $(0,\bar{v})$ for some $\bar{v} > 0$) with $F(0)=0$.

We have the standard notion of Bayesian incentive compatibility:\newline

\noindent{\bf Bayesian Incentive Compatibility}: For each  $v\in\mathbb{R}_+^{\mathbb{N}}$, and each agent $i\in\mathbb{N}$ and bid $u_i\geq 0$ we have
$$E_{\mu(\cdot|v_i)}[v_ix_i^{\tilde{\mathcal{N}}}(v_i,v_{-i})-p^{\tilde{\mathcal{N}}}_i(v_i,v_{-i})]\geq E_{\mu(\cdot|v_i)}[v_ix_i^{\tilde{\mathcal{N}}}(u_i,v_{-i})-p_i^{\tilde{\mathcal{N}}}(u_i,v_{-i})].$$
However, our example will also satisfy the stronger notion of dominant-strategy incentive compatibility from before.

Moreover, we want the mechanism to be Bayesian Sybil-proof.
We consider a rather strong Bayesian notion of Sybil-proofness, which is similar to the one in \cite[]{false_name_vcg_bayesian,mazorra2023}.  The notion requires that creating Sybils does not benefit a bidder in expectation if no-one else uses a Sybil in equilibrium and reports their true valuation.\newline

\noindent {\bf Bayesian Sybil-proofness}: For each $v\in\mathbb{R}_+^{\mathbb{N}}$, each $i\in\mathbb{N}$ and all $J\subset\mathbb{N
}$ and $u\in\mathbb{R}_+^{J}$
\begin{equation*}
    E_{\mu(\cdot|v_i)}[v_ix_i^{\mathcal{N}}(v)-p_i^{\mathcal{N}}(v)]\geq E_{\mu(\cdot|v_i)}\left[\sum_{j\in (J\setminus\mathcal{N})\cup i}v_ix_j^{\mathcal{N}\cup J}(v_i,u_{J\setminus\mathcal{N}},v_{-i})-p_j^{\mathcal{N}\cup J}(v_i,u_{J\setminus\mathcal{N}},v_{-i})\right].
\end{equation*}

We consider a lottery mechanism where the item is allocated with equal chances among agents that pay an entry price and the ``ticket price" is increasing in the number of participants. 

\begin{example}\label{example2}
 Let $Z_{n-1}$ be the maximum of an i.i.d sample of size $n-1$ drawn from the distribution $F$. For each integer $n\ge2$, let $y^\star= \min\{y:(\overline{v}-x)f(x)\leq 1-F(x),\,\forall x\in[y,\overline{v}]\}$, define a number $a_n$  
\[
a_n \;=\; \max\{y^\star,E[Z_{n-1}],F^{-1}\left(\frac{n-2}{n-1}\right)\}.
\]
which we call $a_n$ the ``\emph{ticket price}". The {\bf lottery mechanism with increasing ticket price} for $n$ bids is defined as follows:
\begin{itemize}
\item If all or all but one bids are below $a_n$, the item is allocated by a standard second-price auction, and the winner pays the second-highest bid with uniform tie-breaking rule.
\item If at least $m \ge 2$ bids are at least $a_n$, then each of those $m$ bidders wins the item with probability $1/m$, and the winner's payment is $a_n$.\footnote{By risk neutrality, we could equivalently charge all $m$ bidders (and not only the winner) $a_n/m$.}
\end{itemize}

\end{example}
\smallskip

\noindent

Observe that, as long as $y^\star<\overline{v}$,\footnote{Requiring $y^\star < \overline{v}$ means that beyond $y^\star$, i.e., in the upper tail of $F$, we have 
\begin{equation}
    (\overline{v}-x) f(x) \leq 1 - F(x), \quad \forall x \in [y^\star, \overline{v}].
\end{equation}
Equivalently, in terms of the hazard rate $h_F(x)$, this condition is expressed as
\begin{equation}
    h_F(x) \leq \frac{1}{\overline{v}-x}, \quad \forall x \in [y^\star,\overline{v}].
\end{equation}
In other words, sufficiently close to the endpoint $\overline{v}$, the tail of $F$ is no heavier than that of the uniform distribution on $[0,\overline{v}]$.} for each \(n\), this mechanism differs from the second-price auction with positive probability, since $\Pr[v\geq a_n]>0$ as $a_n<\bar{v}$.
Moreover, the mechanism is dominant-strategy- so in particular Bayesian-incentive compatible as it uses a monotonic allocation rule with Myerson payments. Furthermore, the mechanism is non-wasteful, as it almost surely allocates the item to a buyer, and it is symmetric. Moreover, we can show that under mild conditions, the mechanism is Bayesian Sybil-proof. Before we establish Sybil-proofness, we highlight two important properties of this mechanism. First, observe that if at least two agents bid at or above $a_n$, then the item is allocated by a lottery among those qualifying agents at the posted price $a_n$. In contrast, a second-price auction would allocate the item with probability $1$ to the highest of these agents at a price equal to the second-highest qualifying bid. Hence, with non-negligible probability (specifically $1 - a_n^{\,n-1}$), the outcome and payment under our mechanism differ from those of a standard second-price auction. 

A second notable aspect is that when multiple bidders qualify (i.e.\ have $v_j \ge a_n$), the good is allocated randomly at a fixed price among them. This often yields higher ex-interim utility than a second-price auction, as we will establish in Lemma~\ref{lemma:monotone}.

To see how frequently a lottery with at least two agents occurs, let $X$ be the number of agents whose valuations are at least $a_n$. In the special case that each valuation of bidders is drawn independently from the uniform distribution on the unit interval, the random variable $X$ follows a binomial distribution with parameters $n$ and $(1 - a_n)$. Thus the probability that at least two agents qualify is
\[
\Pr(X \ge 2)
\;=\;
1 \;-\; a_n^{\,n} \;-\; n\,\bigl(1 - a_n\bigr)\,a_n^{\,n-1}\approx 1-2/e.
\]
In other words, with probability $\approx 1-2/e$ the mechanism runs a lottery among at least two agents.

To prove Sybil-Proofness, using affine transformations, it is sufficient to show it for $F$ strictly increasing with support $[0,1]$. The following Lemma is proved in the Appendix.
\begin{lemma}\label{lemma:monotone} Suppose agents have valuations drawn from a distribution $F$ strictly increasing with support $[0,1]$.
 When all agents report truthfully, any agent $i$'s ex-interim expected  payoff conditional on knowing the number of submitted bids under the lottery with increasing ticket price mechanism is decreasing in the ticket price $a\in[y^\star,1]$. Moreover, any agent's ex-interim expected payoff is at least as high as in a standard second-price auction.
\end{lemma}
With the lemma, we can prove that the mechanism is Bayesian Sybil-proof.
\begin{theorem}\label{thm:general}
The lottery mechanism with increasing ticket prices is Bayesian Sybil-proof. In particular, if $y^\star<\overline{v}$, then there exists a mechanism that is payment normalized, non-wasteful, symmetric, incentive compatible, and Bayesian Sybil-proof that is non-equivalent to the second-price auction.
\end{theorem}

\begin{proof}
We establish the result for the case that the number of bidders is known by the bidders. A fortiori this also establishes the result if the number of bidders follows a non-degenerate prior.
An agent with value $v < a_n$ cannot profit by splitting into multiple identities: the mechanism behaves like a second-price auction, where multi‐bid strategies do not help since the second-price auction is Sybil-proof.  Hence consider an agent with value $v \geq a_n$.

Suppose the agent creates $k = k_1 + k_2$ Sybils,\footnote{Subsequently, we use the term ``creates $k$ Sybils" for the case where the agent bids from $k$ accounts (his ``original" account and $k-1$ ``fake" accounts).} with $k_1$ Sybils bidding a value greater or equal $a_{n-1 + k}$ and $k_2$ Sybils bidding a value below $a_{n-1 + k}$.  Because $\{a_n\}$ is increasing and $a_n\geq y^\star$, Lemma~\ref{lemma:monotone} rules out profitable manipulations when $k_1=1$ and $k_2\ge 2$.  

Now consider $k_1 \ge 2$.  One checks (as in the proof of Lemma~\ref{lemma:monotone}) that the agent's ex‐interim payoff from this Sybil strategy is
\[
U(v,k_1,k_2,a_{n-1 + k}):
\;=\;
\bigl(v - a_{n-1 + k}\bigr)
\sum_{m=0}^{n-1}
  \binom{n-1}{m}\,
  \bigl(1 - F(a_{n-1 + k})\bigr)^m\,
  \bigl(F(a_{n-1 + k})\bigr)^{\,n-1 - m}
  \frac{k_1}{m + k_1}.
\]
Since $0\le \frac{k_1}{m+k_1}\le 1$ for every $m\in\{0,\dots,n-1\}$, we have
\begin{equation*}
0 \le \sum_{m=0}^{n-1}\binom{n-1}{m}\bigl(1-F(a_{n-1+k})\bigr)^{m}F(a_{n-1+k})^{\,n-1-m}\frac{k_1}{m+k_1}\le 1 .
\end{equation*}
If $a_{n-1+k}>v$, then $U(v,k_1,k_2,a_{n-1+k})\le 0$. Otherwise, if $a_{n-1+k}\le v$, then
\begin{equation*}
\begin{aligned}
U(v,k_1,k_2,a_{n-1+k})
&\le (v-a_{n-1+k})\sum_{m=0}^{n-1}\binom{n-1}{m}\bigl(1-F(a_{n-1+k})\bigr)^{m}F(a_{n-1+k})^{\,n-1-m}\\
&= v-a_{n-1+k},
\end{aligned}
\end{equation*}
since the binomial weights sum to $1$. Hence, in either case,
\begin{equation*}
U(v,k_1,k_2,a_{n-1+k})\le \max\{v-a_{n-1+k},0\}\le v-a_n,
\end{equation*}
where the last inequality uses that $\{a_t\}$ is increasing (so $a_{n-1+k}\ge a_n$) and that we are in the case $v\ge a_n$.
On the other hand,
\begin{equation*}
\begin{aligned}
v-a_n
&\le (v-\mathbb{E}[Z_{n-1}\mid Z_{n-1}\le v])\,\Pr[Z_{n-1}\le v]\\
&\le U(v,a_n),
\end{aligned}
\end{equation*}
where $U(v,a_n)$ denotes the interim payoff under the lottery mechanism with ticket price $a_n$ for the bidder with valuation $v$.
The second inequality is deduced by Lemma~\ref{lemma:monotone} since $(v-\mathbb{E}[Z_{n-1}\mid Z_{n-1}\le v])\,\Pr[Z_{n-1}\le v]$ is the ex-interim expected payoff in the second-price auction.
The first inequality holds since
\begin{equation*}
\begin{aligned}
v-(v-\mathbb{E}[Z_{n-1}\mid Z_{n-1}\le v])\,\Pr[Z_{n-1}\le v]
&= v\,\Pr[Z_{n-1}>v]+\mathbb{E}[Z_{n-1}\mid Z_{n-1}\le v]\,\Pr[Z_{n-1}\le v]\\
&\le \mathbb{E}[Z_{n-1}]
\le a_n,
\end{aligned}
\end{equation*}
where the last inequality is by construction of $a_n$.

Combining the two displays yields $U(v,k_1,k_2,a_{n-1+k})\le U(v,a_n)$.
Hence Sybil strategies yield lower expected payoffs, proving the mechanism is Bayesian Sybil‐proof.
\end{proof}

\section{The Revelation Principle for Sybils}\label{sec:revelation_sybils}
In principle, it might not be bad in itself that participants in a mechanism use Sybils, as long as the mechanism achieves good outcomes. However, it turns out that the possibility to sybil fundamentally constrains what outcomes can be achieved by \emph{any} mechanism. This is a variant of the ``revelation principle". If an arbitrary symmetric indirect mechanism implements an objective in such a way that bidders might use Sybils as part of their strategies, then there is necessarily a corresponding direct mechanism achieving the same objective in a Sybil-proof way. We formally establish this revelation principle in this section, thereby justifying our focus on direct mechanisms in the previous sections.

We define (indirect) mechanisms with variable population to be families $(\mathcal{M}^{\mathcal{N}},x^{\mathcal{N}},p^{\mathcal{N}})_{\mathcal{N}\subset\mathbb{N},|\mathcal{N}|<\infty}$ of mechanisms for each finite $\mathcal{N}\subset\mathbb{N}$ set of players, where $\mathcal{M}^{\mathcal{N}}=\times_{i\in\mathcal{N}}\mathcal{M}^{\mathcal{N}}_i$ are strategy (or message) spaces, $x^{\mathcal{N}}:\mathcal{M}^{\mathcal{N}}\to\Delta(\mathcal{N})$ is an allocation function and $p^{\mathcal{N}}:\mathcal{M}^{\mathcal{N}}\to\mathbb{R}_+^{\mathcal{N}}$ is a payment function. We call a mechanism symmetric if there is a message space $\mathcal{M}$ such that for all finite $\mathcal{N}\subset\mathbb{N}$ and all $i\in\mathcal{N}$ we have $\mathcal{M}^{\mathcal{N}}_i=\mathcal{M}$ and for all permutations $\pi$ of $\mathcal{N}$ we have $x^{\mathcal{N}}_{\pi(i)}((s_{\pi(i)}))=x^{\mathcal{N}}_i(s)$ and $p^{\mathcal{N}}_{\pi(i)}((s_{\pi(i)}))=p^{\mathcal{N}}_i(s)$. The {\bf corresponding mechanism with Sybils} for a symmetric mechanism $(\mathcal{M},\{x^{\mathcal{N}},p^{\mathcal{N}}\}_{\mathcal{N}\subset\mathbb{N},|\mathcal{N}|<\infty})$ is given by $(\tilde{\mathcal{M}},\{\tilde{x}^{\mathcal{N}},\tilde{p}^{\mathcal{N}}\}_{\mathcal{N}\subset\mathbb{N},|\mathcal{N}|<\infty})$ with $\tilde{\mathcal{M}}=\bigcup_{n\in\mathbb{N}}\mathcal{M}^n$ and $$\tilde{x}_i^{[N]}(s)=\sum_{j=1 }^{dim(s_i)} x^{[\tilde{N}_N]}_{\tilde{N}_i+j}(vec(s) ),\quad\tilde{p}_i^{[N]}(s)=\sum_{j=1 }^{dim(s_i)} p^{[\tilde{N}_N]}_{\tilde{N}_i+j}(vec(s) ),$$ where $\tilde{N}_i:=\sum_{j=1}^{i-1}dim(s_j)$ and $vec(s)$ denotes the concatenation of $s_1,s_2,\ldots,s_N$. We extend the outcome function to sets $\mathcal{N}$ that are not of the form $[N]=\{1,\ldots,N\}$ by requiring the Sybil mechanism to be symmetric.

Generally, subtleties arise around the players reasoning about the true number of bidders. Subsequently, we assume that bidders cannot condition their message on the number of messages the Sybil mechanism receives (because they do not know it when sending the message to the mechanism). For the dominant strategy implementation case, we can then operate in a ``belief-free" model: sending message $m_i(v_i)$ should be weakly-dominant for a bidder $i$ with value $v_i$, for all possible sets of participants in the Sybil mechanism. For the Bayesian setting bidders will also hold beliefs about the number (or more generally set of) bidders that participate in the Sybil mechanism, as we further discuss below.

A social welfare function for a variable population model is a function $f:\mathcal{D}\to\Delta(\mathbb{N})\times\mathbb{R}_+^{\mathbb{N}}$ with $f_i(v_1,\ldots)=(0,0)$ for $v_i=0$ where $\mathcal{D}\subseteq c_{00}(\mathbb{R}_+)$ is a subset of the set of $\mathbb{R}_+$-valued sequences with finitely many non-zero elements. The first element of $f_i(v)$ corresponds to the probability with which bidder $i$ gets the item, and the second element corresponds to the payment bidder $i$ makes. A social welfare function yields a direct mechanism where $\mathcal{M}=\mathbb{R}_+$ and $(x^{\mathcal{N}}(v),p^{\mathcal{N}}(v))=f(v|_{\mathcal{N}},0|_{\mathbb{N}\setminus\mathcal{N}})|_{\mathcal{N}}$.

We say that a (Sybil) mechanism $(\tilde{\mathcal{M}},\tilde{x},\tilde{p})$ implements  $f$ in weakly dominant strategies iff there is a message profile $(m_i(\cdot))_{i\in \mathbb{N}}$ such that for each $v=(v_i)\in c_{00}(\mathbb{R}_+)$,
\begin{enumerate}
\item for each $i\in \text{supp}(v)$ and each finite $\mathcal{N}\subset\mathbb{N}$ with $i\in\mathcal{N}$ reporting $m_i(v_i)$ to mechanism $(\tilde{\mathcal{M}},\tilde{x},\tilde{p}) $ is a weakly dominant strategy,
\item  for each $v$ we have $f(v)=(\tilde{x}^{\text{supp}(v)}(m(v)),\tilde{p}^{\text{supp}(v)}(m(v)))$.
\end{enumerate}
We have the following proposition which is similar to an observation by~\cite[]{false_name}:
\begin{proposition}[Revelation Principle for Sybils in Dominant strategies]
    Let $(\mathcal{M},x,p)$ be a symmetric mechanism and $(\tilde{\mathcal{M}},\tilde{x},\tilde{p})$  be the corresponding mechanism with Sybils. If $(\tilde{\mathcal{M}},\tilde{x},\tilde{p})$ implements $f$ in weakly dominant strategies, then the direct mechanism $f$ is Sybil-proof and incentive compatible.
\end{proposition}
\begin{proof}
Suppose not. Then there is a bidder $i$ with value $v_i$ that is better off from lying or from sybilling in the direct mechanism at some set $\mathcal{N}$ of participants who submit bids $v_{-i}$. In the first case, if reporting $u_i$ instead of $v_i$ would yield higher pay-off if the other bids in the mechanism are $v_{-i}$, then the standard revelation principle argument applies: sending message $m_i(u_i)$ would yield a higher payoff in the Sybil mechanism than sending $m_i(v_i)$ in case the other participants in the Sybil mechanism send messages $m_j(v_j)$. This contradicts $m_i(v_i)$ being a dominant strategy if $i$ has value $v_i$. In the second case, reporting $v_i$ and creating a Sybil $i'$ bidding $u_{i'}$ yields a higher payoff  if the other bids in the mechanism are $v_{-i}$. But then sending message $m=(m_i(v_i),m_{i'}(u_{i'}))$ to the Sybil mechanism would yield higher payoff than sending message $m_i(v_i)$ when other bidders send messages $m_j(v_j)$, contradicting $m_i(v_i)$ being a weakly dominant strategy for bidder $i$ in the Sybil mechanism.
\end{proof}
To formulate the Bayesian version of a Sybil revelation strategy we add a prior to the model which is a probability measure $\pi$ over $c_{00}(\mathbb{R}_+)$.\footnote{A natural case would for example be the i.i.d. private value setting with an unknown number of bidders where we have a probability  $\lambda(N)$ of $N$ players participating in the Sybil mechanism for each of which the value is drawn i.i.d. from a distribution $F$.}

We say that a (Sybil) mechanism $(\tilde{\mathcal{M}},\tilde{x},\tilde{p})$ implements $f$ in Bayes-Nash equilibrium iff there is  is a strategy profile $(m_i(\cdot))_{i\in \mathbb{N}}$ such that for each $v\in c_{00}(\mathbb{R}_+)$,
\begin{enumerate}
\item for each $i\in \text{supp}(v_i)$ and each $m'\in M$ we have 
$$E[v_ix_i(m(v))-p_i(m(v))|v_i]\geq E[v_ix_i(m_i',m_{-i}(v_{-i}))-p_i(m_i',m_{-i}(v_{-i}))|v_i]$$ 
\item   for each $v$ we have $f(v)=(\tilde{x}^{\text{supp}(v)}(m(v)),\tilde{p}^{\text{supp}(v)}(m(v)))$.
\end{enumerate}
\begin{proposition}[Revelation Principle for Sybils for Bayes-Nash equilibrium]
    Let $(\mathcal{M},x,p)$ be a symmetric mechanism and $(\tilde{\mathcal{M}},\tilde{x},\tilde{p})$  be the corresponding mechanism with Sybils. If $(\tilde{\mathcal{M}},\tilde{x},\tilde{p})$ implements $f$ in Bayes-Nash equilibrium, then the direct mechanism $f$ is Bayesian incentive compatible and Bayesian Sybil-proof.
\end{proposition}
\begin{proof}
Suppose not. Then there is a bidder $i$ with value $v_i$ that is better off from lying or from sybilling in the direct mechanism given that everyone else bids truthfully and does not Sybil. In the first case, if reporting $u_i$ instead of $v_i$ would yield higher expected pay-off given truthful bidding of the other bidders, then messaging $m_i(u_i)$ instead of $m_i(v_i)$ to the Sybil mechanism if other bidders follow the equilibrium bidding strategy would be a better response contradicting $m_i(\cdot)$ being a best response. Similarly, if creating a Sybil $i'$ and bidding $u_{i'}$ from it makes bidder $i$ better of in expectation in the direct mechanism, then sending message $(m_i(v_i),m_{i'}(u_{i'}))$ to the Sybil mechanism is a better response to the other bidders' strategies than sending message $m_i(v_i)$ contradicting the fact that $(m_j(\cdot))$ is a Bayes-Nash equilibrium of the direct mechanism.
\end{proof}

\section{Extensions}
We briefly discuss other extensions of our baseline model and whether our main (im)possibility result extends to these settings.
\subsection{Multiple Goods with Unit-Demand}
 A natural question is whether Theorem~\ref{maintheorem} extends to the case of multiple goods. With multiple goods, however, there might not even be an efficient, symmetric, Incentive Compatible, and Sybil-proof direct mechanism if some bidders demand multiple units.\footnote{The VCG mechanism is for example not always Sybil-proof, see \cite[]{false_name}.} One simple setting where there are multiple non-wasteful, symmetric, Incentive Compatible, and Sybil-proof direct mechanisms is the case of unit demand bidders. 

More precisely, we assume there are $m$ identical copies of an item, and each agent has \textit{unit-demand} valuations: an agent has a value for receiving for up to one unit but no additional value for receiving more than one. 

We extend the notion of \textbf{non‑wastefulness} as follows. For all finite $\mathcal{N}\subset\mathbb{N}$ ,
\begin{enumerate}
\item
    $0\leq x^{\mathcal{N}}_i(v)\leq 1\quad\text{for all } i\in\mathcal{N}$,
\item defining $n:=|\mathcal{N}|$, we have
    $$\begin{cases}

        \displaystyle\sum_{i\in\mathcal{N}} x^{\mathcal{N}}_i(v)=m, & \text{if } n\ge m,\\[6pt]

        x_i(v)=1 \;\text{for every } i\in{\mathcal{N}}, & \text{if } n<m.

    \end{cases}$$
    \end{enumerate}
Thus, we always allocate all $m$ items (so none go unused), and we never allocate more than one item to any bidder, since doing so would be wasteful given unit‑demand valuations. We refer to this as \textbf{non‑wastefulness for unit‑demand bidders}.

The {\bf uniform-price mechanism} is defined as follows:
\label{def:uniform-price}
Let $n:=|\mathcal{N}|$ be the number of submitted bids.
For the bid profile $v$, write the order statistics as  
$v_{(1)}\ge v_{(2)}\ge\dots$.  Define the  clearing price
\[
p(v)\;:=\;
\begin{cases}
  v_{(m+1)}, & \text{if } n>m,\\[4pt]
  0,         & \text{if } n\le m.
\end{cases}
\]
The mechanism awards one item to each of the bidders whose bid is at least $v_{(m)}$.  
        If several bidders tie at the cutoff, we break ties uniformly at random so that exactly $m$ items are assigned when $n\ge m$.
        Every allocated bidder pays the uniform price $p(v)$; all other bidders pay~$0$.

As the uniform price mechanism is a natural generalization of the second price auction to this setting, it is natural to ask whether
it is the unique mechanism that is symmetric, Incentive Compatible, Sybil‑proof, and non‑wasteful for unit‑demand bidders.

The answer is no, as we demonstrate with the following counterexample.

\begin{example}
Let $m$ be the number of items. Let $v$ be the bid profile and, w.l.o.g., assume the bids are ordered in descending order. In case that $n\le m$ or $v_{m+1}< v_1/m$ allocate the items as a uniform‑price auction. Otherwise, allocate the $m$ items by selecting, uniformly at random, a subset of $m$ agents from the set $[m+1]$ and giving one item to each selected agent. Consequently, every agent receives an item with probability $\tfrac{m}{m+1}$.
\end{example}

Clearly this allocation rule is monotonic, therefore, the Myerson payments make the mechanism Incentive Compatible. The mechanism is non‑wasteful, and symmetric by construction.

The following statement is proved in the appendix.
\begin{proposition}\label{MultiUnit}
The mechanism is Sybil‑proof.
\end{proposition}
\subsection{Super-additive Utility}\label{sec:superadditive}
 Our main theorem can be extended to a broader class of preferences. Specifically, suppose each agent \(i\) has a private type \(\theta_i\). If the mechanism allocates an amount \(x_i\) of the good to agent~\(i\) with type $\theta_i$ and charges a payment \(p_i\), then agent~\(i\)'s utility is 
\[
  U(\theta_i,x_i, p_i) \;=\; v(\theta_i, x_i) \;-\; p_i,
\]
where \(v: \mathbb{R}_+ \times [0,1] \to \mathbb{R}_+\) is a strictly increasing in $\mathbb R_{>0}\times (0,1]$ and twice differentiable function with $v(\theta_i,0)=0$ and $v(0,x_i)=0$. Thus, as in previous sections, agents are ex-ante symmetric (otherwise, discussions of Sybils make little sense), but their interim utility depends on a private signal. We make two assumptions on the valuation function. We require value to be convex in types:
\newline

\noindent \label{A4}{\bf Type-convexity:} for every $x\in[0,1]$ the function $\theta_i\mapsto v(\theta_i,x)$ is convex.\newline

The second assumption is
\newline

\noindent
{\bf Superadditive interim utility:} for every \(\theta_i\ge0\) the function
          \(x\mapsto\partial_{1}v(\theta_i,x)\) is superadditive on \([0,1]\):
          \[
           \partial_{1}v(\theta_i,x)+\partial_{1}v(\theta_i,y)\;\le\;\partial_{1}v(\theta_i,x+y),
            \quad 0\le x,y,\;x+y\le1.          \]
\newline
Note that Myerson payments in this setting are given by
$$p(\theta)=\int_0^{x_i(\theta)}\partial_{2}v(\theta_i,x_i)dx_i=v(\theta_i,x_i(\theta))-\int_0^{\theta_i}\partial_{1}v(\tilde\theta_i,x_i)d\tilde\theta_i.$$
Thus, the interim utility under Myerson payments for bidder $i$ with type $\theta_i$ is 
$$U(\theta_i,x(\theta),p(\theta))=\int_0^{\theta_i}\partial_{1}v(\tilde\theta_i,x)d\tilde\theta_i.$$
Super-additivity of $\partial_{1}v(\tilde\theta_i,x)$ therefore implies that the interim utility is super-additive in consumption which will turn out to make sybilling a worthwhile strategy in mechanisms other than auctions. Type-convexity on the other hand guarantees that the interim utility is increasing in types which is needed for implementation (in dominant strategies).

\label{A4}

 Here are some examples that are part of the class of preferences:
\begin{example}
	Let \(\theta_i\in\mathbb R_{+}\) and \(x_i\in[0,1]\).
	\begin{enumerate}
		\item[] \textbf{Power family}  
		      \[
		          v(\theta_i,x_i)=\theta_i^{\alpha}x_i^{\beta}, 
		          \qquad \alpha\geq1,\;\beta\geq1 .
		      \]
		\item[] \textbf{Exponential–power family}  
		      \[
		          v(\theta_i,x_i)=\theta_i^{\alpha}\bigl(e^{\kappa x_i}-1\bigr), 
		          \qquad \alpha\geq1,\;\kappa>0 .
		      \]
  \end{enumerate}
\end{example}

The notions of payment normalization, symmetry, non-wastefulness, incentive compatibility and Sybil-proofness can be adapted in a straightforward way to this setting.
We then have the following extension of Theorem~\ref{maintheorem}, which is proved in the appendix:
\begin{proposition}\label{thm:superadditive}
For valuations that are type-convex and induce superadditive interim utility, the generalized second-price auction with symmetric tie-breaking is the only direct mechanism that simultaneously satisfies payment normalization, non-wastefulness, symmetry, incentive compatibility, and Sybil-proofness.
\end{proposition}

\section{Conclusion}
Our main theorem can be interpreted as a trilemma for the design of non-wasteful mechanism in the presents of Sybils: any such mechanism is either not Incentive-Compatible, not Sybil-proof or centralizing.  In light of our main theorem,
if we aim to design a practical mechanism that does not always assign the good to the highest-value bidder, then we need to relax Sybil-proofness, Incentive Compatibility, or Non-wastefulness. We discuss some possible directions to relax these axioms. 

In some contexts, being wasteful seems to be very undesirable,\footnote{This is for example the case in our running example of block proposing, as blockchains shouldn't produce empty blocks.} whereas in others it could be tolerable. Sybil-proofness is easier to satisfy when the probability of not allocating the good to anyone is increasing in the number of bidders, as Examples~\ref{example} shows. However, the mechanism in Example~\ref{example} can hardly be of practical use, as the probability of wasting the good converges to $1$ when there are many bidders. Thus, we ask:
\begin{question}
    Does there exist a Symmetric, Incentive Compatible, and Sybil-proof mechanism such that the probability of allocating the good to a bidder is larger than some positive constant?
\end{question}
Practically, relaxing Incentive Compatibility is a natural direction to explore. Our impossibility result implies that no non-trivial distributional objectives can be achieved by any Sybil-proof mechanism. However, that does not exclude the possibility of approximately satisfying objectives. Thus, we ask
\begin{question}
   How well can (worst-case) equilibria of Sybil-proof mechanisms, such as the proportional mechanism, approximate distributional objectives?
\end{question}
In many practical scenarios, creating Sybils is cheap but not entirely free \cite[]{mazorra2023}. If we assume that creating Sybils has a small constant cost for all bidders, will there be a substantially larger class of mechanisms satisfying all the desirable properties? In particular, we are interested in which distributional objectives we can achieve.
\begin{question}
   When creating Sybils is costly, what is the class of non-wasteful, symmetric, IC and Sybil-proof mechanisms?
\end{question}
In Section~\ref{Bayes}, we have seen that relaxing Sybil-proofness to a Bayesian version extends the set of possible mechanisms. There are arguably many sensible notions of Bayesian Sybil-proofness that can be explored. Moreover, it would be interesting to classify the design space of Bayesian Sybil proof mechanisms.

\begin{question}
What is the class of non-wasteful, symmetric, Bayesian IC and Bayesian Sybil-proof mechanisms?
\end{question}

\bibliographystyle{RM.bst}
\bibliography{refs}

\begin{thebibliography}{}

\bibitem[Brill et~al., 2016]{brill2016}
Brill, M., Freeman, R., Conitzer, V., and Shah, N. (2016).
\newblock False-name-proof recommendations in social networks.

\bibitem[Budish and Bhave, 2023]{budish2023}
Budish, E. and Bhave, A. (2023).
\newblock Primary-market auctions for event tickets: Eliminating the rents of “bob the broker”?
\newblock {\em American Economic Journal: Microeconomics}, 15(1):142--170.

\bibitem[Chen et~al., 2019]{chen2019}
Chen, X., Papadimitriou, C., and Roughgarden, T. (2019).
\newblock An axiomatic approach to block rewards.
\newblock In {\em Proceedings of the 1st ACM Conference on Advances in Financial Technologies}, pages 124--131.

\bibitem[Gafni et~al., 2020]{false_name_vcg_bayesian}
Gafni, Y., Lavi, R., and Tennenholtz, M. (2020).
\newblock {VCG} under sybil (false-name) attacks - {A} bayesian analysis.
\newblock In {\em The Thirty-Fourth {AAAI} Conference on Artificial Intelligence, {AAAI} 2020, The Thirty-Second Innovative Applications of Artificial Intelligence Conference, {IAAI} 2020, The Tenth {AAAI} Symposium on Educational Advances in Artificial Intelligence, {EAAI} 2020, New York, NY, USA, February 7-12, 2020}, pages 1966--1973. {AAAI} Press.

\bibitem[Gafni and Tennenholtz, 2023]{false_name_combinatorial}
Gafni, Y. and Tennenholtz, M. (2023).
\newblock Optimal mechanism design for agents with {DSL} strategies: The case of sybil attacks in combinatorial auctions.
\newblock In Verbrugge, R., editor, {\em Proceedings Nineteenth conference on Theoretical Aspects of Rationality and Knowledge, {TARK} 2023, Oxford, United Kingdom, 28-30th June 2023}, volume 379 of {\em {EPTCS}}, pages 245--259.

\bibitem[Gafni and Yaish, 2024]{transaction_fee_mechanisms}
Gafni, Y. and Yaish, A. (2024).
\newblock Barriers to collusion-resistant transaction fee mechanisms.
\newblock In Bergemann, D., Kleinberg, R., and Sab{\'{a}}n, D., editors, {\em Proceedings of the 25th {ACM} Conference on Economics and Computation, {EC} 2024, New Haven, CT, USA, July 8-11, 2024}, pages 1074--1096. {ACM}.

\bibitem[Gershkov et~al., 2013]{ECTA2}
Gershkov, A., Goeree, J.~K., Kushnir, A., Moldovanu, B., and Shi, X. (2013).
\newblock On the equivalence of bayesian and dominant strategy implementation.
\newblock {\em Econometrica}, 81(1):197--220.

\bibitem[Huberman et~al., 2021]{huberman2021}
Huberman, G., Leshno, J.~D., and Moallemi, C. (2021).
\newblock Monopoly without a monopolist: An economic analysis of the bitcoin payment system.
\newblock {\em The Review of Economic Studies}, 88(6):3011--3040.

\bibitem[Leshno and Strack, 2020]{Leshno2020}
Leshno, J.~D. and Strack, P. (2020).
\newblock Bitcoin: An axiomatic approach and an impossibility theorem.
\newblock {\em American Economic Review: Insights}, 2(3):269--286.

\bibitem[Manelli and Vincent, 2010]{ECTA1}
Manelli, A.~M. and Vincent, D.~R. (2010).
\newblock Bayesian and dominant-strategy implementation in the independent private-values model.
\newblock {\em Econometrica}, 78(6):1905--1938.

\bibitem[Mazorra and Della~Penna, 2023]{mazorra2023}
Mazorra, B. and Della~Penna, N. (2023).
\newblock The cost of sybils, credible commitments, and false-name proof mechanisms.
\newblock {\em arXiv preprint arXiv:2301.12813}.

\bibitem[Messias et~al., 2023]{messias2023}
Messias, J., Yaish, A., and Livshits, B. (2023).
\newblock Airdrops: Giving money away is harder than it seems.
\newblock {\em arXiv preprint arXiv:2312.02752}.

\bibitem[Myerson, 1981]{myerson1981}
Myerson, R.~B. (1981).
\newblock Optimal auction design.
\newblock {\em Mathematics of operations research}, 6(1):58--73.

\bibitem[\"{O}z et~al., 2024]{oz}
\"{O}z, B., Sui, D., Thiery, T., and Matthes, F. (2024).
\newblock Who wins ethereum block building auctions and why?
\newblock In {\em Proceedings of the 5th Conference on Advances in Financial Technologies (AFT)}.

\bibitem[Wagman and Conitzer, 2008]{wagman2008}
Wagman, L. and Conitzer, V. (2008).
\newblock Optimal false-name-proof voting rules with costly voting.
\newblock In {\em AAAI}, volume~8, pages 190--195.

\bibitem[Yokoo et~al., 2004]{false_name}
Yokoo, M., Sakurai, Y., and Matsubara, S. (2004).
\newblock The effect of false-name bids in combinatorial auctions: new fraud in internet auctions.
\newblock {\em Games Econ. Behav.}, 46(1):174--188.

\end{thebibliography}

\appendix

\section{Lemmas for the proof of Theorem \ref{thm:general}}

\begin{lemma}\label{lemma:Uva_equiv}
Let \(F\) be a distribution function on \([0,\infty)\), and fix \(0 \le a \le v\) such that $F(a)<1$. 
Define
\[
U(v,a)
\;=\;
\int_0^a F(x)^{n-1}\,dx
\;+\;
(v-a)\sum_{m=0}^{n-1}
  \binom{n-1}{m}\,
  \bigl[1-F(a)\bigr]^m\,
  F(a)^{\,n-1-m}\,
  \frac{1}{m+1}.
\]
Then
\[
U(v,a)
\;=\;
\int_{0}^{a} \bigl[F(x)\bigr]^{n-1}\,dx
\;+\;
(v - a)\,\frac{1 - \bigl[F(a)\bigr]^n}{\,n\,\bigl[1 - F(a)\bigr]}\,.
\]
\end{lemma}
\begin{proof}
We only need to simplify
\begin{equation*}
S \;:=\;\sum_{m=0}^{n-1}\binom{n-1}{m} (1-F(a))^m F(a)^{\,n-1-m}\frac{1}{m+1}.
\end{equation*}
Use the identity $\frac{1}{m+1}\binom{n-1}{m}=\frac{1}{n}\binom{n}{m+1}$ to get
\begin{align*}
S
&= \frac{1}{n}\sum_{m=0}^{n-1}\binom{n}{m+1} (1-F(a))^m F(a)^{\,n-1-m}
= \frac{1}{n}\sum_{k=1}^{n}\binom{n}{k} (1-F(a))^{k-1} F(a)^{\,n-k} \\
&= \frac{1}{n(1-F(a))}\sum_{k=1}^{n}\binom{n}{k} (1-F(a))^{k} F(a)^{\,n-k}
= \frac{1}{n(1-F(a))}\Bigl(1 - F(a)^n\Bigr)
= \frac{1-F(a)^n}{n(1-F(a))}.
\end{align*}
Substituting back into the definition of $U(v,a)$ yields the claim. 
\end{proof}
\subsection{Proof of Lemma \ref{lemma:monotone}}\label{appendix:lemma}
\begin{proof}
We first prove the following:\\\newline
\noindent\textbf{Claim}: $\frac{
n \left( 1 + (n-1) F(a)^{n} - n F(a)^{n-1} \right)
\left( (v-a) f(a) - (1 - F(a)) \right)
}{
\left( n (1 - F(a)) \right)^2
}\leq 0$ for $a\in [y^\star,1]$ and $v\in [a,1]$.\\

Since $a\in [y^\star,1]$ and $v\in [a,1]$, $(v-a) f(a) - (1 - F(a))\leq 0$ by definition of $y^\star$. Therefore, to prove the claim, is sufficient to show that $1+(n-1)F(a)^n-nF(a)^{n-1}\geq0$.
Define the function
\begin{equation*}
g(x) = 1 + (n-1)x^n - n x^{n-1}, \quad \text{for } x \in [0,1].
\end{equation*}
Evaluating at the endpoints, we obtain
\begin{equation*}
g(0) = 1, \quad g(1) = 1 + (n-1) - n = 0.
\end{equation*}

Next, computing the derivative,
\begin{equation*}
g'(x) = (n-1)n x^{n-1} - n(n-1) x^{n-2} = n(n-1) x^{n-2} (x - 1).
\end{equation*}
Since \( x - 1 \leq 0 \) for \( x \in [0,1] \), it follows that \( g'(x) \leq 0 \) in this range, meaning \( g(x) \) is decreasing. Since \( g(0) = 1 \) and \( g(1) = 0 \), we conclude that \( g(x) \geq 0 \) for all \( x \in [0,1] \), which proves the claim.\newline

Let $i$ be an agent with valuation $v$, and define $K=\{\,j \neq i \,\mid\, v_j \ge a\}$. Let $U(v,a)$ denote agent $i$'s expected interim payoff with the lottery mechanism with ticket price $a$. We want to show that if $a'\geq a$, then $U(v,a)\geq U(v,a')$. First, we compute $U(v,a)$. We consider two cases:

\smallskip
\textbf{Case 1:} $v < a$.  
The agent's expected payoff is equivalent to the payoff of the second-price auction, i.e. the expected payoff is
\[
(v-E[Z_{n-1}\mid Z_{n-1}\leq v])\Pr[Z_{n-1}\leq v] = \int_0^v F(x)^{n-1}dx
\]
independent of $a$. Thus, no change occurs as $a$ crosses above $v$.

\smallskip
\textbf{Case 2:} $v \ge a$.  
Now agent $i$ always qualifies ($i\in K$).  Let $Y=|K|-1$ be the number of other bidders whose bids exceed $a$, so $Y$ is $\mathrm{Binomial}(n-1,\,1-F(a))$.  

If $Y=0$, the mechanism is effectively a second-price auction with $n$ bidders, among whom $i$ is guaranteed highest if the others are all below $a$. The payoff in this scenario is 
\begin{equation*}
    F(a)^{n-1}(v-E[Z_{n-1}\mid Z_{n-1}\leq a])=F(a)^{n-1}(v-a)+\int_0^a F(x)^{n-1}dx
\end{equation*}

If $Y\ge1$, the mechanism assigns the item at random among the $Y+1$ qualifying bidders, charging $a$.  The agent's expected payoff is
\[
\sum_{m=1}^{n-1}
  \Pr[Y=m]\,
  \frac{v - a}{\,m+1\,}
\;=\;
(v-a)\sum_{m=1}^{n-1}
  \binom{n-1}{m}
  (1-F(a))^m
  F(a)^{n-1-m}
  \frac{1}{m+1}.
\]
Adding the two parts (for $Y=0$ and $Y\ge1$) yields
\[
U(v,a)
\;=\;
\int_0^a F(x)^{n-1}dx
\;+\;
(v-a)\sum_{m=0}^{n-1}
  \binom{n-1}{m}
  (1-F(a))^m\,F(a)^{\,n-1-m}
  \frac{1}{m+1}.
\]
By Lemma~\ref{lemma:Uva_equiv} the sum equals 
\[
U(v,a)
= \int_{0}^{a} F(x)^{n-1} \,dx
+ (v - a) \frac{1 - \bigl[F(a)\bigr]^{n}}{n \bigl[1 - F(a)\bigr]}.
\]
 Observe that
\begin{align*}
    \frac{\partial U(v,a)}{\partial a} &= F(a)^{n-1}-\frac{1-F(a)^n}{n(1-F(a))}+(v-a)f(a)\frac{-n^2F(a)^{n-1}(1-F(a))+n(1-F(a)^n)}{(n(1-F(a))^2}\\
    &=\frac{
n \left( 1 + (n-1) F(a)^{n} - n F(a)^{n-1} \right)
\left( (v-a) f(a) - (1 - F(a)) \right)
}{
\left( n (1 - F(a)) \right)^2
}
\end{align*}
The previous expression is non-positive for $a\in[y^\star,1]$ and $v\in [a,1]$, by the Claim.

Putting both claims together, $\frac{\partial U(v,a)}{\partial a}\leq 0$ for $a\in[y^\star,1]$ and $v\in[a,1]$. Therefore, $U(v,a)$ is monotone non-increasing in this region. In particular, if $v\geq a_n$, 
\[
U(v,a)\geq U(v,v)= \int_0^v F(x)^{n-1}dx.
\]

\end{proof}

\section{Proof of Proposition~\ref{MultiUnit}}

\begin{proof}
Let $v_1,\dots,v_{n-1}$ be the bids reported by the other agents, and assume w.l.o.g. that they are
in descending order and $n\ge m+1$. Consider agent $i=n$ with value $v$ who (when using one account)
submits a bid $v$. The Myerson payment of agent $i$ is
\[
p_i(v,v_{-i})=v\cdot x_i(v,v_{-i})-\int_0^{v} x_i(t,v_{-i})\,dt.
\]

\smallskip
\noindent\textbf{Allocation rule.}
Example~6.1 runs the uniform-price auction iff $v_{(m+1)}< v_{(1)}/m$ (order statistics of the full
profile), and otherwise runs the $(m+1)$-lottery among the top $m+1$ bids. Fixing $v_{-i}$, this yields
the following interim allocation rule for bidder $i$. Let $r:=\max\{v_{m+1},v_1/m\}$.
Then
\[
x_i(t,v_{-i})=\begin{cases}
1, & \text{if }\max\{t,v_1\} > m\,v_m \text{ and } t\ge v_m,\\[3pt]
\frac{m}{m+1}, & \text{if }\max\{t,v_1\}\le m\,v_m \text{ and } t\ge r,\\[3pt]
0, & \text{otherwise.}
\end{cases}
\]
(In particular, when $v_{m+1}\le t < v_m$, the lottery branch requires $t\ge v_1/m$.)

If $v<r$ then $x_i(v,v_{-i})=0$ and $p_i(v,v_{-i})=0$, so we assume $v\ge r$.

\smallskip
\noindent\textbf{Myerson payments.}

\textbf{Case~1:} $v\ge m v_{m}$.
\begin{enumerate}
\item If $v\le v_1$, then the payment is exactly $v_m$.
\item If $v>v_1$. If $v_1\ge m v_{m}$, then the payment is $v_{m}$. Otherwise, the payment is
$\tfrac{m}{m+1}v_{m}+\tfrac{m}{m+1}\max\{v_1/m,v_{m+1}\}$.
\end{enumerate}

\textbf{Case~2:} $v< m v_{m}$.
\begin{enumerate}
\item If $v_1\leq m v_m$, then (since $v\ge r$) the payment is $\tfrac{m}{m+1}\max\{v_1/m,v_{m+1}\}$.
\item If $v_1> m v_m$, the payment is $v_m$ in case $v\ge v_m$ and $0$ otherwise.
\end{enumerate}

\smallskip
\noindent\textbf{Sybil deviations.}
The result is trivial if $n\le m$ (then everyone receives an item and pays $0$), so assume $n>m$.

Consider a unit-demand agent with value $v$ who deviates by submitting $k\ge 1$ bids
$b^1,\dots,b^k\ge 0$ (accounts), with no restriction that any $b^\ell$ equals $v$.
Let $b^{(1)}\ge b^{(2)}\ge \cdots \ge b^{(k)}$ be the bids sorted.

A first observation is that no agent has an incentive to use Sybil bids below the largest $m+1$ bids (including his other Sybil bids) since those will just increase payments without increasing the allocation probability. A second observation is that no agent is incentivized to bid from more than two Sybils. A third Sybil will not increase the probability of getting an item, and just will increase payments. When using two Sybils, if one bids less than the $m+1$‑st biggest bid, that Sybil will not get allocated the item, and can only increase the payment of the other Sybil, making this deviation also not profitable. Therefore, we may assume that agents bid with two Sybils with one bid larger than $v_{m}$. Now fix a two-bid deviation $(a,b)$ with $a\ge b$.

If the deviation induces the uniform-price auction, it is dominated by a single bid.
Indeed, removing the lower bid $b$ can only weakly decrease the $(m+1)$-st order statistic and leaves
the top bid unchanged, so if the profile with $(a,b)$ triggers the auction branch then the profile with only
$a$ also triggers the auction branch. The lower bid $b$ cannot increase the probability of receiving a unit
above $1$ (unit demand) and can only (weakly) increase payments. Therefore such a deviation is weakly
dominated by a single-bid deviation with bid $a$, which cannot be profitable by incentive compatibility.

Thus it remains to consider two-bid deviations that induce the lottery.
In that case both bids have allocation probability $q:=\frac{m}{m+1}$. Moreover, for each of the two
accounts, its opponent set contains the $m$ bids $v_1,\dots,v_m\ge v_m$ and the other account’s bid
(which is $\ge v_m$ in this lottery-relevant case), so the $(m+1)$-st highest opponent bid is at least $v_m$.
From the payment formula in Case~2.1 (applied to the relevant opponent profile of each account),
each account’s expected payment is at least $q v_m$. Hence the \emph{total} expected payment across the
two accounts is at least $2q v_m=\frac{2m}{m+1}v_m$, and the deviating utility is at most
\begin{equation}\label{equation:B}
    v-\frac{2m}{m+1}v_m.
\end{equation}

\smallskip
\noindent\textbf{Comparison with truthful utility.}
We now show truthful bidding yields utility at least \eqref{equation:B} in all cases; this implies no Sybil deviation is
profitable.

First consider $v<v_1$.
\begin{enumerate}
\item If $v_1\ge m v_m$, the mechanism coincides with the uniform-price auction for the relevant
range and Sybil deviations are ruled out by the dominance argument above (auction-branch deviations
reduce to a single bid, which is unprofitable by IC).
\item If $v_1<m v_m$, then truthful utility $U$ (for a single bid) equals
\[
U=\frac{m}{m+1}\Bigl(v-\max\{v_1/m,v_{m+1}\}\Bigr)\ge \frac{m}{m+1}(v-v_m),
\]
since $\max\{v_1/m,v_{m+1}\}\le v_m$ in this case. Because $v<v_1\le m v_m$, we have $v\le m v_m$,
which is equivalent to
\[
v-\frac{2m}{m+1}v_m\le \frac{m}{m+1}(v-v_m)\le U.
\]
Together with \eqref{equation:B}, this rules out profitable Sybil deviations.
\end{enumerate}
Next consider $v\ge v_1$.
\begin{enumerate}
\item The cases $v_1\ge m v_m$ or $v_1<m v_m$ and $v<m v_m$ are analogous to the previous case.
\item If $v_1<m v_m$ and $v\ge m v_m$, then by the payment computation in Case~1.2.,
truthful utility is
\[
v-\tfrac{m}{m+1}v_m-\tfrac{m}{m+1}\max\{v_1/m,v_{m+1}\}\ge v-2\tfrac{m}{m+1}v_m.
\]
Any Sybil deviation that induces the auction branch is dominated by a single-bid deviation (hence not
profitable by IC), and any Sybil deviation that induces the lottery has utility bounded by \eqref{equation:B}. Thus no
Sybil deviation is profitable.
\end{enumerate}
\end{proof}

\section{Proof of Theorem~\ref{thm:superadditive}}
\label{sec:superadditive}
The proof relies on two lemmas. First we have the following version of Myerson's lemma which is straightforward to establish so that we omit the proof: 
\begin{lemma}[Myerson-lemma for general valuations]
\label{thm:myerson-unbounded}
Suppose $v$ is strictly increasing in $\mathbb R_{>0}\times (0,1]$ and twice differential function with $v(\theta,0)=0$ and $v(0,x)=0$ and we have type-convexity and superadditive interim utility. Then, a mechanism \((x,p)\) is Incentive Compatible if and only if, for every
agent \(i\) and every profile \(\theta_{-i}\),
\begin{enumerate}
    \item the allocation \(\theta_i\mapsto x_i(\theta_i,\theta_{-i})\) is non-decreasing on \([0,\infty)\); and
    \item payments satisfy the envelope formula
          \[
              p_i(\theta_i,\theta_{-i})
              \;=\;
              v(\theta_i,x_i(\theta_i,\theta_{-i}))
              \;-\;
              \int_{0}^{\theta_i} \partial_1 v(z,x_i(z,\theta_{-i}))\,dz,
              \qquad \theta_i\ge 0.
          \]
\end{enumerate}
\end{lemma}
Second we have the following properties of the value function:
\begin{lemma}\label{lem:SC+SA}
If $v$ is strictly increasing in $\mathbb R_{>0}\times (0,1]$ and twice differentiable function with $v(\theta,0)=0$ and $v(0,x)=0$ and we have type-convexity and superadditive interim utility then:
\begin{enumerate}
    \item \emph{$\partial_1 v(\theta,x)>0$ for every $\theta,x>0$.} 
    \item \emph{Single-crossing:}\;  
          \(v(\theta',x')-v(\theta',x) \geq v(\theta,x')-v(\theta,x)\) whenever \(\theta'>\theta\) and \(x'>x\);
    \item \emph{Allocation superadditivity:}\;  
          \(v(\theta,x)+v(\theta,y)\le v(\theta,x+y)\) for every \(\theta\ge0\) and \(x,y\ge0\) with \(x+y\le1\).
\end{enumerate}
\end{lemma}

\begin{proof}
\noindent\emph{(1)} Assume, for a contradiction, that
\[
\partial_{1}v(\theta_{0},x)=0 \qquad\text{for some } \theta_{0}>0,\;x>0.
\]
Because \(v\) is convex in the first argument
the function $\theta\mapsto$\(\partial_{1}v(\theta,x)\) is
non-decreasing. On the other hand \(\partial_{1}v(0,x)=0\), so
\[
\partial_{1}v(\theta,x)=0 \quad\text{for every }\theta\in[0,\theta_{0}].
\]
Integrating gives
\[
v(\theta_0,x)-v(0,x)=\int_{0}^{\theta_0}\partial_{1}v(t,x)\,dt=0
\]
hence \(v(\theta_0,x)=v(0,x)\).  This contradicts the assumption that $\theta\mapsto v(\theta,x)$ is strictly increasing. 

\noindent\emph{(2)}  
Fix \(x'>x\) and define \(\phi(\theta):=v(\theta,x')-v(\theta,x)\).  
Then \(\phi'(\theta)=\partial_1 v(\theta,x') -\partial_1 v(\theta,x)\).  
Because \(\partial_1 v(\theta,0)=0\) and \(x\mapsto \partial_1 v(\theta,x)\) is superadditive, it is non-decreasing on \((0,1]\); hence \(\phi'(\theta)\geq0\) for all \(\theta>0\).  
Thus \(\phi\) is non-decreasing, giving the desired inequality for \(\theta'>\theta\).

\noindent\emph{(3)}  
By the properties of $v$ and the fundamental theorem of calculus,
\begin{equation}\label{eq1}
  v(\theta,x)=\int_{0}^{\theta} \partial_1v(t,x)\,dt
  \quad\text{for all }(\theta,x).
\end{equation}
And so, for every $x,y\in [0,1]$ such that $x+y\leq1$,
\begin{equation*}
    v(\theta,x)+v(\theta,y) =  \int_0^{\theta} [\partial_1v(t,x)+\partial_1v(t,y)]dt \leq \int_0^{\theta} \partial_1 v(t,x+y)dt = v(\theta,x+y).
\end{equation*}
\end{proof}

Now we are ready to prove Theorem \ref{thm:superadditive}. The proof follows the same structure as the proof for the linear valuations. In the following, however, we will assume the stronger notion of Sybil-proofness, that is, the agent can employ an arbitrary finite number of Sybil bids.
\begin{lemma} $U_1(\theta_1,\theta_2)= v(\theta_1,1)-v(\theta_2,1)$  and $x_1(u,v)=1$ if $\theta_1\geq \theta_2$, $\theta_1,\theta_2\in\mathbb R_+$.
\end{lemma}
\begin{proof} First, we show that $U_1(\theta_1,\theta_2)\geq v(\theta_1,1)-v(\theta_2,1)$. Suppose that the value of Bidder $1$ is $\theta_1$ and the value of Bidder $2$ is $\theta_2$. Suppose the Bidder $1$ unilaterally deviates and uses $n-1$ Sybils, bidding $\theta_2$ with each one of them. By symmetry, the allocation for each Sybil is $x_i(\theta_{2[n]})=\frac{1}{n}$. Since the mechanism is Incentive Compatible and payments are normalized, it is, in particular, individually rational, and hence, the payment per Sybil is at most $v(\theta_2,1/n)$. Therefore, the utility of the unilateral deviation is at least $v(\theta_1,\frac{n-1}{n})-(n-1)\,v(\theta_2,\frac{1}{n})$. Since the mechanism is Sybil-proof, 
\begin{equation*}
    U_1(u,v)\geq v\left(\theta_1,\frac{n-1}{n}\right )-(n-1)v\left(\theta_2,\frac{1}{n}\right).
\end{equation*}
Since $v(\theta_2,\cdot)$ is superadditive, $(n-1)v(\theta_2,\frac{1}{n})\leq v(\theta_2,1-1/n)$. By continuity of $v$, 
\begin{equation*}
    U_1(u,v)\geq \limsup_{n\rightarrow+\infty} \left [v\left(\theta_1,\frac{n-1}{n}\right)-v\left(\theta_2,1-\frac{1}{n}\right)\right]= v(\theta_1,1)-v(\theta_2,1).
\end{equation*}
Now, we show that $U_1(\theta_1,\theta_2)\leq v(\theta_1,1)-v(\theta_2,1)$.

On one side, we have $U(\theta_1,\theta_2)=\int_0^{\theta_1} \partial_1 v(z,x_1(z,\theta_2))dz$, therefore
\begin{align*}
\int_{0}^{\theta_1}U(\theta_1,\theta_2)\,d\theta_2
   &=\int_{0}^{\theta_1}\!\!\int_{0}^{\theta_1}
        \partial_{1}v(z,x_1(z,\theta_2))\,dz\,d\theta_2\\[6pt]
   &=
        \iint_{\{0\le \theta_2<z\le \theta_1\}}
            \partial_{1}v(z,x(z,\theta_2)\bigr)\,dz\,d\theta_2
        \;+\;
        \iint_{\{0\le z<\theta_2\le \theta_1\}}
            \partial_{1}v(z,x(z,\theta_2))\,dz\,d\theta_2\\[6pt]
   &=
        \iint_{\{0\le \theta_2<z\le \theta_1\}}
            \partial_{1}v(z,x(z,\theta_2))\,dz\,d\theta_2
        \;+\;
        \iint_{\{0\le \theta_2<z\le \theta_1\}}
            \partial_{1}v(\theta_2,x(\theta_2,z))\,dz\,d\theta_2\\[6pt]
   &=
        \iint_{\{0\le \theta_2<z\le \theta_1\}}
            \!\Bigl[
                \partial_{1}v(z,x(z,\theta_2))
               +\partial_{1}v(\theta_2,1-x(z,\theta_2))
              \Bigr]\,dz\,d\theta_2\\[6pt]
   &\le
        \iint_{\{0\le \theta_2<z\le \theta_1\}}
            \partial_{1}v(z,1)\,dz\,d\theta_2\\[8pt]
   &=
        \int_{0}^{\theta_1}\!
            \Bigl[\int_{0}^{z}d\theta_2\Bigr]\partial_1 v(z,1)\,dz
        \;=\;
        \int_{0}^{\theta_1}z\,\partial_1 v(z,1)\,dz\\[8pt]
   &=
        \Bigl[z\,v(z,1)\Bigr]_{0}^{\theta_1}-\int_{0}^{\theta_1}v(z,1)\,dz
        \;=\;
        \theta_1\,v(\theta_1,1)-\int_{0}^{\theta_1}v(\theta_2,1)\,d\theta_2\\[6pt]
   &=
        \int_{0}^{\theta_1}\bigl[v(\theta_1,1)-v(\theta_2,1)\bigr]\,d\theta_2. \\ 
   &\leq \int_0^{\theta_1} U(\theta_1,\theta_2)d\theta_2
\end{align*}
In the first step, we use Myerson lemma to rewrite
\(\int_{0}^{\theta_1} U(\theta_1,\theta_2)\,d\theta_2\) as 
\(\int_{0}^{\theta_1}\!\!\int_{0}^{\theta_1} \partial_1 v\!\bigl(z,x(z,\theta_2)\bigr)\,dz\,d\theta_2\).
Next, we partition the square \([0,\theta_1]\times[0,\theta_1]\) into the two
congruent right--triangles
\(T_{1}=\{(z,\theta_2)\mid 0\le \theta_2<z\le \theta_1\}\) and
\(T_{2}=\{(z,\theta_2)\mid 0\le z<\theta_2\le \theta_1\}\).
Relabelling the variables on \(T_{2}\) and invoking the
identity \(x(z,\theta_2)+x(\theta_2,z)=1\) allow the two contributions to be combined,
yielding the integrand
\(\partial_1 v\bigl(z,x(z,\theta_2)\bigr)+\partial_1 v\bigl(\theta_2,1-x(z,\theta_2)\bigr)\).

For the first the inequality, fix a point with \(\theta_2<z\).
Because \(\theta\mapsto\partial_1 v(\theta,y)\) is non--decreasing, we have
\(\partial_1 v(\theta_2,1-x)\le\partial_1 v(z,1-x)\).
Superadditivity in the second argument then gives
\(\partial_1 v(z,x)+\partial_1 v(z,1-x)\le\partial_1 v(z,1)\),
so the integrand is bounded above by \(\partial_1 v(z,1)\).

Since this upper bound depends only on \(z\), the inner integration in
\(\theta_2\) simply contributes the factor \(z\), leaving
\(\int_{0}^{\theta_1} z\,\partial_1 v(z,1)\,dz\).
An elementary integration by parts
transforms this integral into
\(\int_{0}^{\theta_1}[v(\theta_1,1)-v(\theta_2,1)]\,d\theta_2\). In the last inequality, we use that $U_1(\theta_1,\theta_2)\geq v(\theta_1,1)-v(\theta_2,1)$.

In short, we showed that $U_1(\theta_1,\theta_2) = v(\theta_1,1)-v(\theta_2,1)$ a.s.

Next, we show that $\theta\mapsto U(\theta_1,\theta)$ is non-increasing. Let $\theta_2\geq\theta_2'$, then
\begin{align*}
    U_1(\theta_1,\ \theta_2) &= \int_0^{\theta_1} \partial_1 v(z,x_1(z,\ \theta_2))dz = \int_0^{\theta_1} \partial_1 v(z,1-x_1(\theta_2,z))dz \\
    &\leq \int_0^{\theta_1} \partial_1 v(z,1-x_1(\theta_2',z))dz = \int_0^{\theta_1} \partial_1 v(z,x_1(\theta_2',z))dz \\
    & \leq \int_0^{\theta_1} \partial_1 v(z,x_1(z,\theta_2))dz = U_1(\theta_1,\theta_2')
\end{align*}

Therefore, for every $\theta$, there exists a sequence $\theta_i\uparrow \theta$ such that $U_1(\theta,\theta_i)=v(\theta,1)-v(\theta_i,1)$. So, since $\theta'\mapsto U(\theta,\theta')$ is non-increasing,
\begin{equation*}
    U(\theta,\theta)\leq \liminf_{n\rightarrow +\infty}U(\theta,\theta_i)=0
\end{equation*}
so that $U(\theta,\theta)=0$. However, by Myerson lemma
\begin{equation}\label{eq:utility0}
    0=U_1(\theta,\theta) =\int_0^{\theta} \partial_1 v(z,x_1(z,\theta))dz,
\end{equation}
so that $\partial_1 v(z,x_1(z,\theta))=0$ a.s. for $z\in [0,\theta)$. By Lemma \ref{lem:SC+SA}, if $z,x_{1}(z,\theta)>0$ then $\partial_{1}v(z,x_{1}(z,\theta))>0$. 
Assume, toward a contradiction, that there exists $z_{0}>0$ such that $x_{1}(z_{0},\theta)>0$. 
Because the mapping $z\mapsto x_{1}(z,\theta)$ is non-decreasing, we would have $x_{1}(z,\theta)>0$ for every $z\in[z_{0},\theta)$; hence $\partial_{1}v(z,x_{1}(z,\theta))>0$ on that interval. 
Consequently,
\begin{equation*}
    \int_{0}^{\theta}\partial_{1}v(z,x_{1}(z,\theta))\,dz>0,
\end{equation*}
which contradicts \ref{eq:utility0}. 
Therefore, $x_{1}(z,\theta)=0$ for all $z\in[0,\theta)$.

\end{proof}
In the previous lemma, we have established that for the case of two bidders, the allocation rule always assigns the whole item to the highest value bidder. Next, we show by induction on the number of bidders that in the case of more than two bidders, the item is also allocated to the highest value bidder.

We claim that for any $n \geq 2$ and any $\theta_1' < \theta_1 := \max\{\theta_2, \dots , \theta_n\}$, Bidder $1$ will not
obtain the object if he reports $\theta_1'$, i.e. $x_1(\theta_1', \theta_2, \dots, \theta_n) = 0$. We proceed by induction
on $n$. We already know the base case $n = 2$ is true. Suppose that the claim is true
for $n = k$ and assume towards contradiction that $x_1(\theta_1', \theta_2,\dots, \theta_{k+1}) > 0$ for some $\theta_1'<\theta_1:=\max\{\theta_2,\dots.\theta_{k+1}\}$. 
Assume that the types are labeled so that $\theta_2\geq \theta_3\geq \dots \geq \theta_k \geq \theta_{k+1}$.
Since $\theta\mapsto x_1(\theta,\theta_2,\dots,\theta_{k+1})$ is non-decreasing, by Lemma \ref{lem:SC+SA} this implies that $\partial_1 v(z,x_1(z,\theta_2,\dots,\theta_{k+1}))>0$ for all $\theta\geq \theta_1'$. Consider the scenario where there are $k$ bidders, Bidder $1$ has value $\theta_1$ and Bidder $i$ has value $\theta_i$ for $2\leq i\leq k$. If Bidder $1$ bids truthfully, his utility payoff is
\begin{equation*}
    \int_0^{\theta_1} \partial_1 v(z,x_1(z,\theta_2,\dots,\theta_k))dz=0
\end{equation*}
by induction hypothesis, since $\theta_1:=\max(\theta_2,\dots\theta_{k+1}) = \max(\theta_2,\dots,\theta_{k})$, by labeling. However, if Bidder $1$ deviates by bidding $\theta_1$ himself and
creating a Sybil that bids $\theta_{k+1}$, then his utility payoff is (where $x_i(z):=x_i(z,\theta_{-i})$ and $p_i(z):=p_i(z,\theta_{-i})$ for $i=1,k+1$ and $z\in\mathbb R_+$):
\begin{align*}
    v(\theta_1,x_1(\theta_1)+x_{k+1}(\theta_{k+1})) -p_1(\theta_1)-p_{k+1}(\theta_{k+1}) &\geq v(\theta_1,x_1(\theta_1))-p_1(\theta_1)+v(\theta_1,x_{k+1}(\theta_{k+1}))-p_{k+1}(\theta_{k+1}) \\
    &\geq \int_0^{\theta_1} \partial_1 v(z,x_1(z))dz + \underbrace{v(\theta_{k+1},x_{k+1}(\theta_1))-p_{k+1}(\theta_{k+1})}_{\geq0\text{ By IR}}\\
    &\geq\int_{\theta_1'}^{\theta_1} \partial_1 v(z,x_1(z))dz>0
\end{align*}
which contradicts Sybil-proofness.
To summarize, we have shown that the allocation rule shall reward the item to the
highest bidder (when there is a tie among bids, the symmetry assumption implies
uniform tie-breaking rule) and therefore, by Incentive Compatibility, the mechanism
is the generalized second-price auction with the symmetric tie-breaking rule.

\end{document}